\documentclass[letterpaper, 10 pt, conference]{ieeeconf}  

\IEEEoverridecommandlockouts                              
\usepackage{amsmath} 
\usepackage{amssymb}  
\usepackage{mathtools}
\usepackage{./TLC}
\usepackage{physics} 
\usepackage{blox}
\usepackage{flushend}
\usepackage{tikz,pgfplots}
\usepgfplotslibrary{external}
\usepackage{shellesc}

\usepackage[style=ieee]{biblatex}
\definecolor{redp}{RGB}{252,141,98}
\definecolor{bluep}{RGB}{141,160,203}

\title{\LARGE \bf
Circuit Model Reduction with Scaled Relative Graphs
}

\author{Thomas Chaffey$^{1}$ \and Alberto Padoan$^{2}$%
\thanks{$^{1}$University of Cambridge, Department of Engineering, Trumpington Street,
        Cambridge CB2 1PZ, {\tt\small tlc37@cam.ac.uk}.}
        \thanks{%
        $^{2}$ETH Z\"urich,  Department of Information Technology and Electrical Engineering,  Physikstrasse 3,  8092,  Z\"urich,  Switzerland,  {\tt\small apadoan@ethz.ch}.}
        }
\begin{document}

\maketitle
\thispagestyle{empty}
\pagestyle{empty}

\begin{abstract}
        Continued fractions are classical representations of complex objects
        (for example, real numbers) as sums and inverses of  simpler  objects (for
        example, integers).   The analogy in linear circuit theory is a chain of  series/parallel    one-ports:  the   
        port behavior is a continued fraction containing the port behaviors of its elements.
        Truncating  a   continued fraction is a classical method of
        approximation,   which   corresponds to deleting the circuit
        elements furthest from the port.  We apply this idea to   chains of series/parallel  
        one-ports composed of arbitrary nonlinear  relations.  This gives a model
        reduction method which automatically preserves properties such as incremental
        positivity.  The Scaled Relative
        Graph (SRG) gives a graphical representation of the original and truncated port
        behaviors.  The difference of these SRGs
        gives a bound on the approximation error, which is shown to be competitive
        with existing methods.  
\end{abstract}

\section{Introduction}

Continued fractions are classical in the theory of approximation
\autocite{jones1984continued}, and are closely related to Pad\'e approximants
\autocite{baker1996pade}, which have had a broad impact in areas such as 
theoretical physics \autocite{baker1970pade}, fluid mechanics \autocite{cabannes1976pade},  and control theory~\autocite{kalman1979partial,antoulas1986recursiveness,bultheel2000rational,antoulas2005approximation}.  
Continued fractions also have a long and rich history in linear circuit theory~\cite{newcomb1966linear}.  They have been
used extensively for synthesis and approximation,  beginning in the seminal works of Foster
\autocite{Foster1924},  Cauer \autocite{Cauer1926},   Bott,  Duffin
\autocite{Bott1949},   and  Kalman  \autocite{kalman1979partial},  among others. 
The Cauer normal forms for RC and RL circuits are
continued fractions of transfer functions \autocite{newcomb1966linear}, and  
the truncation of a continued fraction corresponds to deleting
elements from a series/parallel one-port.   The  nonlinear counterpart of this
fruitful circle of ideas, however, is largely unexplored.
In this context,    this paper proposes the truncation of a ``continued fraction'' of nonlinear relations as
a paradigm for model reduction of nonlinear series/parallel one-ports.  

The aim of model reduction is to approximate a complex model by a simpler one,
whilst retaining the important behavior.   In particular, one may require that
properties of the model (such as stability or passivity) are preserved by the
approximation. 
For linear systems,  the literature on the subject is vast, see, for example,
\autocite{antoulas2005approximation} and references therein.   In contrast,  the problem is largely open for
nonlinear systems.  Some progress has been made in~\autocite{scherpen1993balancing}, ~\autocite{astolfi2010model} and~\cite{padoan2022LSMR}, and 
in the recent papers~\autocite{besselink2013model} and~\cite{padoan2021mr}, 
the Lur'e structure is exploited to reduce a nonlinear model,   while preserving incremental dissipativity properties.  A  nonlinear system is represented as a linear time invariant (LTI) state space model in feedback with a static nonlinearity.  The LTI component  is then   approximated
using standard methods, such as  balanced truncation and Hankel norm approximation  \autocite{antoulas2005approximation}.   
Although computationally effective,  this procedure  does not
leverage the structure of the underlying
physical system, and it is difficult, in general, to guarantee that the approximate
system exhibits desired properties, such as positivity (a close relative of passivity
\autocite[Lemma 2, p. 200]{Desoer1975}).  In contrast, we propose that a system be
modelled from the very beginning as an interconnection of physical components, and
the system be approximated by deleting the components which are least important.
Properties which are preserved by physical interconnection, such as positivity
\autocite[$\S5$, Chap. 6]{Desoer1975}, are then naturally retained in the
approximate system. 
The choice of electrical
terminology is purely a matter of preference: series/parallel electrical circuits
have analogies in domains such as mechanics, hydraulics and thermodynamics
\autocite{Smith2002,vanderSchaft2014}.

This paper proposes  the Scaled Relative Graph (SRG) as a tool for
quantifying the errors introduced by an approximation.  The SRG has
recently been introduced in the theory of convex optimization \autocite{Ryu2021}, and
allows simple, graphical proofs of algorithm convergence, and the derivation of tight
convergence bounds \autocite{Huang2020}.  The SRG gives a graphical representation of
the incremental behavior of a nonlinear operator, and generalizes the Nyquist diagram of an LTI
transfer function \autocite{Chaffey2021c}.  Interconnections of operators correspond
to graphical combinations of their SRGs \autocite{Ryu2021}, and applying this
graphical algebra to the study of feedback systems gives rise to  a nonlinear  
Nyquist criterion, which generalizes many existing results on incremental
input/output stability \autocite{Chaffey2021c, Chaffey2022}.  Properties such as
incremental gain and incremental positivity can be read directly from the SRG, and as
such the SRG may be used to measure the error introduced in such quantities.
Plotting an SRG for the error system, that is, the difference between a system and
its approximation, allows us to bound the incremental gain from input to
approximation error.  This bound is shown to compare favourably with other bounds in
the literature.

We begin this paper in Section~\ref{sec:illustrative} with a  motivating  example, 
which illustrates the main ideas.  We then define the model class and propose a truncation
procedure in Section~\ref{sec:truncation}.  Section~\ref{sec:verify} introduces the
SRG, and how it may be used to  certify approximation error bounds.   Equipped  with these tools, we
revisit the example circuit in Section~\ref{sec:example},  and compare our method with
the method described in~\autocite{Besselink2014}.    
Section~\ref{sec:conclusions} concludes the paper with a summary of our main results,
and poses open questions for future research.

\section{A motivating example}\label{sec:illustrative}

The running example of this paper is the circuit illustrated in
Figure~\ref{fig:example_circuit_1}.  $G_{RC}$ is the admittance of an LTI RC filter,
and $R$ is an arbitrary nonlinear resistor $v = R(i)$ which, for all $v_1, v_2, i_1, i_2$, satisfies the
incremental sector bound
\begin{IEEEeqnarray}{rCl}
        \mu \Delta i^2 \leq \Delta i \Delta v \leq \lambda \Delta i^2,\label{eq:sector}
\end{IEEEeqnarray}
for some $0 \leq \mu \leq \lambda$, where $\Delta i = i_1 - i_2$ and $\Delta v = v_1
- v_2$. 
\begin{figure}[ht]
        \centering
        \includegraphics{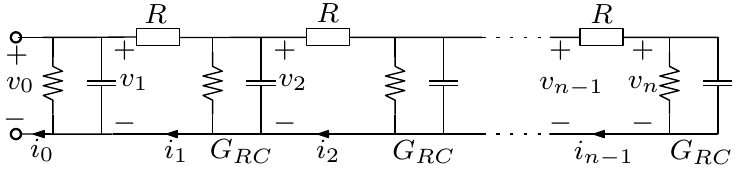}
        \caption{A nonlinear lattice circuit, configured as a one-port.}%
        \label{fig:example_circuit_1}
\end{figure}

The circuit consists of   a chain  $n$ repeated three element units, and an  additional RC filter at the port.  
A first attempt at approximating the circuit might simply be to remove the units
furthest from the port.  This corresponds to truncating a ``continued fraction'' in the circuit elements (to be
made precise in Section~\ref{sec:truncation}).  A better method might be to only
remove the capacitors furthest from the port,
 and resolve the remaining (linear and nonlinear) resistors into a single
nonlinear resistor.  This gives the same continued fraction truncation, with an
additional nonlinear resistance.
  If $R$ is LTI, one can show that the approximation error is always bounded 
by the $H_\infty$ norm of the original circuit's transfer function~\autocite{Srinivasan1997}.  
We generalize this result to the case where $R$ is nonlinear,
using SRGs. 

The SRG of a circuit of length $n$ is shown in Figure~\ref{fig:SRG_final}.  The value
$\lambda_n$ (defined in Section~\ref{sec:example}) bounds the incremental secant gain 
\autocite[$\S$2]{Sontag2006} of the circuit, and we will show that it also bounds the gain
in the error, $\norm{v - \hat v}/\norm{i}$, where $i$ is an input current, $v$ is the
output voltage of the original circuit and $\hat v$ is the output voltage of any
circuit with the last $r < n$ capacitors removed.

\begin{figure}[h]
        \centering
        \includegraphics{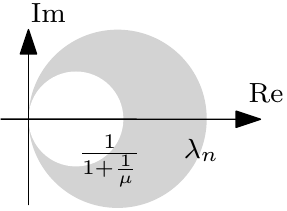}
        \caption{SRG for the current to voltage relation of the circuit of Figure~\ref{fig:example_circuit_1}.}%
        \label{fig:SRG_final}
\end{figure}

\section{Truncating series/parallel one-port circuits} \label{sec:truncation}

\subsection{Circuit elements as relations}
Let $L_2$ denote the set of finite energy signals $u: [0, \infty) \to \R$ such that
        ${\int_{0}^{\infty} |u^2(t)|\mathrm{d}t < \infty.}$ 
The inner product on $L_2$ is defined by
\begin{IEEEeqnarray*}{rCl}
        \bra{u}\ket{y} \coloneqq \int_{0}^{\infty} u(t)y(t) \mathrm{d}t,
\end{IEEEeqnarray*}
which induces the norm $\norm{u} \coloneqq \sqrt{\bra{u}\ket{u}}$.

We consider circuits formed by the parallel and series interconnection of one-port
elements.  A one-port has two terminals, across which a voltage $v$ is measured, and
through which a current $i$ flows.  We assume that these currents and voltages belong
to $L_2$, and a one-port $R$ is described by a relation on $L_2$, that is, a set $R \subseteq L_2\times L_2$ of
ordered voltage/current pairs.  If a one-port is described by a relation from voltage to
current, it is an \emph{admittance}, and if a one-port is described by a relation
from current to voltage, it is an \emph{impedance}.  If $(u, y) \in R$, we write $y
\in R(u)$.

The usual operations on functions can be extended to relations:
\begin{IEEEeqnarray*}{rCl}
        S^{-1} &=& \{ (y, u) \; | \; y \in S(u) \}\\
        S + R &=& \{ (x, y + z) \; | \; (x, y) \in S,\; (x, z) \in R \}\\
        SR &=& \{ (x, z) \; | \; \exists\, y \text{ s.t. } (x, y) \in R,\; (y, z) \in S \}.
\end{IEEEeqnarray*}
The relational inverse always exists, but is not an inverse in the usual sense -- in
particular, it is in general not the case that $R^{-1}R = I$.  If, however, $R$ is an
invertible function, its functional inverse coincides with its relational inverse, so
the notation $R^{-1}$ is not ambiguous. If $R$ is an impedance,   mapping  $i$ to $v$, then $R^{-1}$ is an admittance, mapping $v$ to $i$.

\begin{definition}
        A relation $R$ on $L_2$, mapping $u$ to $y$, is said to be
        \begin{enumerate}
                \item incrementally positive (or monotone) if 
                        $\bra{u_1 - u_2}\ket{y_1- y_2} \geq 0$ 
                        for all $u_1, u_2, y_1 \in R(u_1), y_2 \in R(u_2)$;
                \item $\mu$-input strictly incrementally positive (or $\mu$-coercive) if 
                        $\bra{u_1 - u_2}\ket{y_1- y_2} \geq \mu \norm{u_1 - u_2}^2$ 
                        for all $u_1, u_2, y_1 \in R(u_1), y_2 \in R(u_2)$;
                \item $1/\gamma$-output strictly incrementally positive (or $1/\gamma$-cocoercive) if 
                        $ \gamma \bra{u_1 - u_2}\ket{y_1- y_2} \geq \norm{y_1 - y_2}^2 $ 
                        for all $u_1, u_2, y_1 \in R(u_1), y_2 \in R(u_2)$.  $\gamma$
                        is called the \emph{incremental secant gain}.
                \item $R$ is said to have an incremental gain bound (or Lipschitz
                        constant) of $\lambda$ if
                        $ \norm{y_1- y_2} \leq \lambda \norm{u_1 - u_2} $ 
                        for all $u_1, u_2, y_1 \in R(u_1), y_2 \in R(u_2)$.
        \end{enumerate}
        The incremental secant gain of a system is also an incremental gain bound.
\end{definition}

Incremental positivity is closely related to incremental passivity -- the two are
equivalent for causal operators (this follows from an easy adaptation of the proof of
\autocite[Lemma 2, p. 200]{Desoer1975}).
Examples of incrementally positive circuit elements include resistors with
nondecreasing $i-v$ characteristics and LTI capacitors and inductors
\autocite{Chaffey2021e}.

\subsection{Series/parallel one-port circuits}

A series interconnection of two impedances $R_1$ and $R_2$ defines a new one-port
impedance,
$R_1 + R_2$:
\begin{IEEEeqnarray*}{rCl}
          v \in R_1(i) + R_2(i). 
\end{IEEEeqnarray*}
Likewise, the parallel interconnection of two admittances $G_1$ and $G_2$ defines the
one-port admittance $G_1 + G_2$:
\begin{IEEEeqnarray*}{rCl}
        i \in G_1(v) + G_2(v).
\end{IEEEeqnarray*}
Interconnecting an impedance and an admittance, either in series or in parallel,
requires one of the relations to be inverted. We will assume throughout this paper
that any relations which are added have compatible domains.  For a circuit-theoretic
interpretation of this assumption, see \autocite[Thm. 2]{Chaffey2021e}.

These interconnection rules give rise to the class of one-port circuits consisting of
arbitrary series/parallel interconnections, which have the general form shown in
Figure~\ref{fig:series-parallel-general}  (allowing admittances to be open circuits,
$\{(v, 0)\;|\; v \in \R\}$, and impedances to be short circuits, $\{(i, 0)\;|\; i \in
\R\}$).

\begin{figure}[hb]
        \centering
        \includegraphics{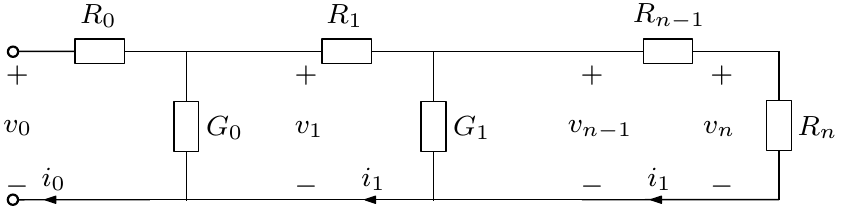}
        \caption{Circuit structure with nested series and parallel interconnections.
        $R_n$ represents an impedance, $G_n$ represents an admittance.}
\label{fig:series-parallel-general}
\end{figure}

The $v-i$ relation of this general circuit is given by
\begin{IEEEeqnarray*}{rCl}
v_0 &=& (R_0 + (G_0 + (\ldots + (R_{n-1} + R_n)^{-1}\ldots
)^{-1})^{-1})(i_0).
\end{IEEEeqnarray*} 

Note that this form generalizes a continued fraction of transfer functions: when all the elements $R_j$ and $G_j$ are LTI,  
taking the Laplace transform gives
\begin{IEEEeqnarray*}{rCl}
v_0(s) = R_0(s) +\cfrac{1}{G_0(s) + \cfrac{1}{\ldots + \cfrac{1}{R_{n-1}(s)+R_{n}(s)}}}\, i_0(s).
\end{IEEEeqnarray*}
When every element is incrementally positive, such circuits are closely related to
the splitting algorithms of monotone operator theory, and may be solved efficiently
using the recently introduced class of \emph{nested splitting algorithms}
\autocite{Chaffey2021e}.

\subsection{Truncated approximate circuits}

Consider the problem of approximating the one-port circuit in
Figure~\ref{fig:series-parallel-general},  whose port behavior 
is given by $v_0=C(i_0)$, by a simpler one-port.
A natural solution is to delete the circuit elements furthest from the port terminals,
as they contribute the least to the port behavior of the circuit.   
This gives a truncated circuit $\hat C$ with $i-v$ relation, defined by 
\begin{IEEEeqnarray*}{lCl}
\hat{v}_0 
 =  (R_0 + (G_0 + (\ldots + (R_{r-1} + R_r)^{-1}\ldots
)^{-1})^{-1})(i_0),\label{eq:truncated_circuit}
\end{IEEEeqnarray*}
where $r < n$.
This procedure corresponds to truncating the
continued fraction of the circuit.  The relation from current to voltage has been
chosen arbitrarily, and it is straightforward to verify that the truncation of
$C^{-1}$ is $\hat{C}^{-1}$.  We will also consider the case where the final impedance
$R_r$ is modified to some $\hat{R}_r$ -- for example, the lumped resistance which
remains when only capacitors are removed from the circuit in
Figure~\ref{fig:example_circuit_1}.

In the case that all the circuit elements $R_j$, $G_j$ are incrementally positive,
both the original and truncated circuits are automatically incrementally positive --
this follows from the preservation of incremental positivity under series and
parallel interconnections \autocite[Prop. 1]{Chaffey2021e}.  In the case that the
circuit elements have stronger positivity properties, these are also preserved in the
truncated circuit, as shown in the following proposition.

\begin{proposition}\label{prop:strong_preservation}
        Consider  the circuit in Figure~\ref{fig:series-parallel-general}.  
        Suppose that each admittance $G_j$ is input-strictly incrementally positive,
        and each impedance $R_j$ is output-strictly incrementally positive.  Then the
        circuit is input-strictly incrementally positive from voltage to current, and 
        output-strictly incrementally positive from current to voltage.
\end{proposition}

\begin{proof}
        The proof follows from induction, and the following basic results (see, for
        example, \autocite[Chap. 2]{Ryu2021a}).  Let $A$ and $B$ be relations on an arbitrary
        Hilbert space.  Then:
        \begin{enumerate}
                \item If $A$ and $B$ are input-strictly incrementally positive, then
                        $A + B$ is input-strictly incrementally positive.
                \item If $A$ and $B$ are output-strictly incrementally positive, then $A + B$ is output-strictly
                        incrementally positive.
                \item If $A$ is input-strictly incrementally positive, $A^{-1}$ is
                        output-strictly incrementally positive.
                \item If $A$ is output-strictly incrementally positive, $A^{-1}$ is
                        input-strictly incrementally positive.\qedhere
        \end{enumerate}
\end{proof}

As the series/parallel structure of a circuit is preserved as elements are removed,
Proposition~\ref{prop:strong_preservation} shows that truncation preserves strict
incremental positivity.  In the following section, we will develop several
numerical estimates of the accuracy of an approximation, using the circuit's SRG.

\section{Graphical truncation errors}\label{sec:verify}
We begin this section with a brief
overview of the theory of SRGs.  For a full treatment, we refer the reader to
\textcite{Ryu2021}.

\subsection{Scaled Relative Graphs}

The SRG of an operator is a region in the extended complex plane from which the dynamic
properties of the operator can be easily read.  We define the SRG formally as follows.

The angle between $u, y \in L_2$ is defined as
\begin{IEEEeqnarray*}{rCl}
        \angle(u, y) \coloneqq \acos \frac{\bra{u}\ket{y}}{\norm{u}\norm{y}} \in [0,
        \pi). 
\end{IEEEeqnarray*}

Let $R \subseteq L_2 \times L_2$.  Given $u_1, u_2 \in
L_2$, $u_1 \neq u_2$, we define the set of complex numbers $z_R(u_1, u_2)$ by
\begin{IEEEeqnarray*}{rCl}
        \left\{\frac{\norm{y_1 - y_2}}{\norm{u_1 - u_2}} e^{\pm j\angle(u_1 -
u_2, y_1 - y_2)}\middle|\; y_1 \in R(u_1), y_2 \in R(u_2) \right\}.
\end{IEEEeqnarray*}
If $u_1 = u_2$ and there exist corresponding
outputs $y_1 \in R(u_1), y_2 \in R(u_2), y_1 \neq y_2$, then
$z_R(u_1, u_2)$ is defined to be $\{\infty\}$.  If $R$ is single valued at $u_1$,
$z_R(u_1, u_1)$ is the empty set.

The \emph{Scaled Relative Graph} (SRG) of $R$ over $L_2$ is then given by
\begin{IEEEeqnarray*}{rCl}
        \srg{R} \coloneqq \bigcup_{u_1, u_2 \in\, L_2}  z_R(u_1, u_2).
\end{IEEEeqnarray*}

\begin{proposition}\label{prop:finite_gain}
        The SRG of a relation belongs to one of the regions illustrated below if
                and only if the relation obeys the corresponding input/output
        property. Clockwise from top left: finite incremental gain,
$1/\gamma$-output strict incremental positivity, $\mu$-input strict incremental
positivity.%
\begin{center}
        \includegraphics{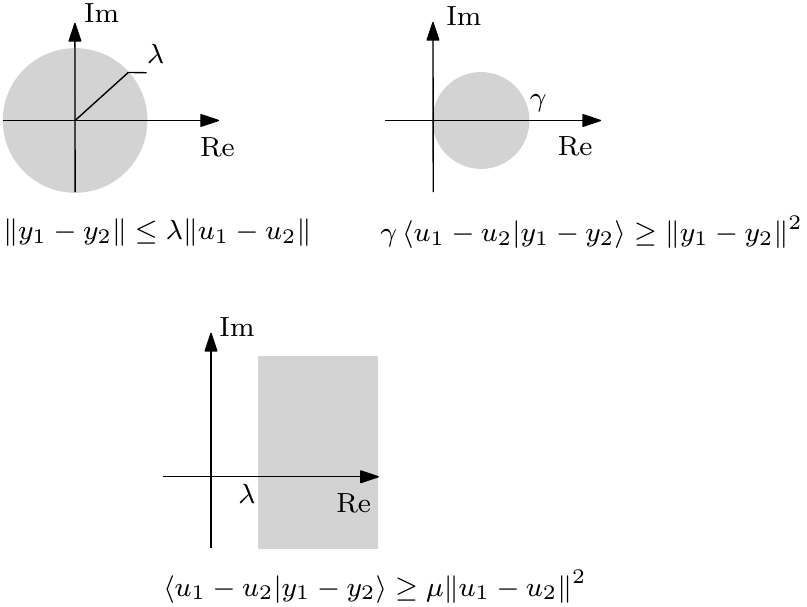}
\end{center}
\end{proposition}

\subsection{SRGs of series/parallel one-ports}

Connecting elements in series and parallel involves adding and inverting their
relations.  In this section, we describe the corresponding graphical operations on
their SRGs.

If $C, D \subseteq \C$, we define the
operation $C + D$ to be the Minkowski sum of $C$ and
$D$, that is,
\begin{IEEEeqnarray*}{rCl}
        C + D \coloneqq \{c + d\, | \, c \in C, d \in D\}.
\end{IEEEeqnarray*}

We define inversion in the  extended  complex plane by $re^{j\omega} \mapsto (1/r)e^{j\omega}$.
This maps points outside the unit circle to the
inside, and vice versa. The points $0$ and $\infty$ are exchanged under inversion. The complex conjugate would normally be taken; this is left
out for convenience, and has no effect as the SRG is symmetric about the real axis.

Define the line segment between $z_1, z_2 \in \C$ as $[z_1, z_2] \coloneqq \{\alpha
                                z_1 + (1 - \alpha) z_2\, |\, \alpha \in [0, 1]\}$.
                                A region $G \subseteq \mathbb{C}$ is said to
satisfy the \emph{chord property} if $z \in G$ implies $[z,
\bar z] \subseteq G$.  If $A$ is a relation, we denote by $\overline{\srg{A}}$ any
region in $\mathbb{C}$ such that $\srg{A} \subseteq \overline{\srg{A}}$ and
$\overline{\srg{A}}$ satisfies the chord property.

\begin{proposition}\label{prop:negation}
        If $A$ is an operator, then
                $\srg{-A} = -\srg{A}$.
        
\end{proposition}

\begin{proposition}\label{prop:inversion}
        If $A$ is an operator, then
                $\srg{A^{-1}} = (\srg{A})^{-1}$.
\end{proposition}

\begin{proposition}\label{prop:summation}
        Let $A$ and $B$ be relations whose SRGs are bounded. Then
                        $\srg{A + B} \subseteq \srg{A} +
                        \overline{\srg{B}}$.
\end{proposition}
Unbounded SRGs can be allowed by setting $\srg{A + B} = \{\infty\}$ if
$\srg{A} = \varnothing$ and $\infty \in \srg{B}$.

\subsection{Graphical truncation errors for series/parallel one-ports}

To evaluate the error introduced by truncating a circuit, we can compute a bounding SRG for the error relation, $C - \hat C$, which
maps $u$ to $e \coloneqq y - \hat y$.  It
follows from Proposition~\ref{prop:finite_gain} that the maximum modulus of $\srg{C -
\hat C}$ bounds the incremental error gain,
\begin{IEEEeqnarray*}{rCl}
        \sup_{u_1, u_2 \in L_2, u_1 \neq u_2}\frac{\norm{e_1 - e_2}}{\norm{u_1 - u_2}}.
\end{IEEEeqnarray*}
If this quantity is bounded, the error relation is continuous on $L_2$: small changes
in the input result in small changes in the error.
Under the assumption that $0 \in C(0)$ and $0 \in \hat C(0)$, the incremental error
 gain,  in turn,    bounds
\begin{IEEEeqnarray*}{rCl}
        \sup_{u \in L_2, \norm{u} \neq 0} \frac{\norm{y - \hat y}}{\norm{u}}.
\end{IEEEeqnarray*}
If this quantity is bounded, the error relation is bounded on $L_2$: bounded inputs
result in bounded errors.

Using the SRGs of the original and truncated circuits, we can furthermore measure the
error in various dynamic properties, such as incremental gain and positivity.
Figure~\ref{fig:secant_estimate}, for example, shows how the error in the secant gain
can be measured from the original and truncated SRGs.

\begin{figure}[hb]
        \centering
        \includegraphics{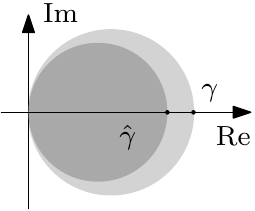}
        \caption{Suppose the light grey region bounds the SRG of the original
        circuit, and the dark grey region bounds the SRG of the truncated circuit.
The original circuit has a secant gain of $\gamma$, and the truncated circuit has a
secant gain of $\hat \gamma$.  The error $\gamma - \hat \gamma$ in the truncated
secant gain is the distance between the two marked points on the real axis.}%
        \label{fig:secant_estimate}
\end{figure}

\section{Example revisited}\label{sec:example}

Armed with the graphical tools of the previous section, we revisit the example of
Section~\ref{sec:illustrative}.  We begin by deriving an SRG for the  circuit in
Figure~\ref{fig:example_circuit_1},    for an arbitrary number of units $n$.  Recall
that $R$ is an arbitrary nonlinear resistor which satisfies the incremental sector
bound \eqref{eq:sector},
and suppose the capacitor and linear resistor both have unit value, $C = G = 1$.  The
SRGs of $R$ and $G_{RC}$ are illustrated below (following \autocite[Thm. 4, Prop.
9]{Chaffey2021c}).

\begin{center}
        \includegraphics{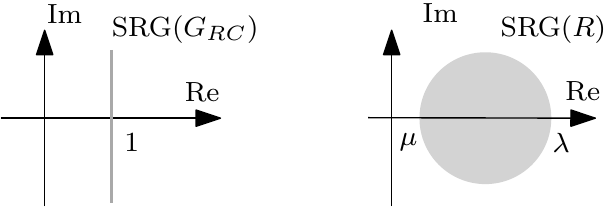}
\end{center}

We then apply the SRG sum and inversion rules (Propositions~\ref{prop:summation}
and~\ref{prop:inversion}) to obtain the SRG for the $i-v$ relation of a circuit with
$n = 1$, shown below (incidentally, we also obtain an SRG for the $v-i$ relation).

\begin{center}
        \includegraphics{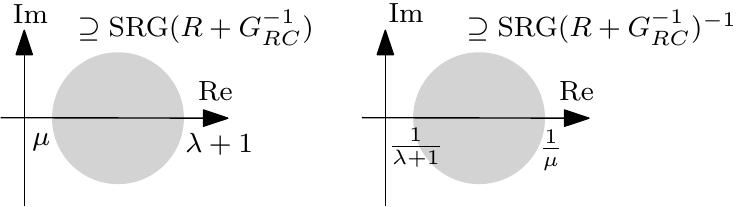} 
\end{center}

\begin{center}
        \includegraphics{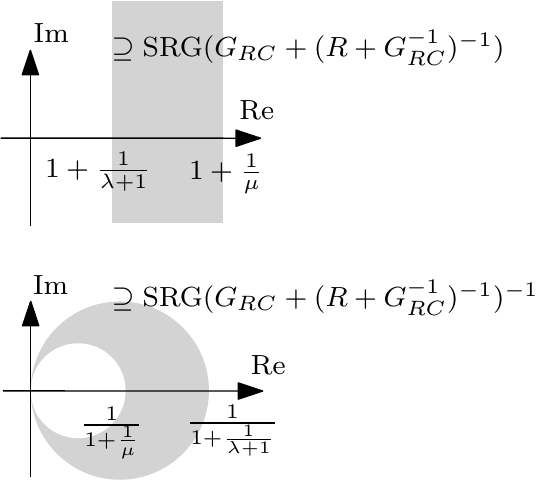} 
\end{center}

Carrying on with this procedure, we obtain the following SRG for a circuit with $n$
units.

\begin{center}
        \includegraphics{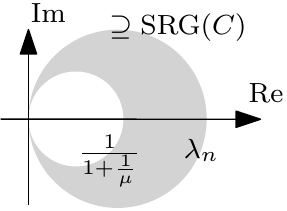}
\end{center}

$\lambda_n$ is defined recursively\footnote{If $\lambda = 1$,   $\lambda_n \to 1/\phi = 2/(1+\sqrt{5})$,  the inverse of the   golden
ratio, as $n\to \infty$.} by
\begin{IEEEeqnarray*}{rCl}
\lambda_n &=&   \begin{cases}
                \frac{1}{1 + \frac{1}{\lambda + 1}} & n = 1\\
                \frac{1}{1 + \frac{1}{\lambda + \lambda_{n-1}}} & n > 1.
                \end{cases}
\end{IEEEeqnarray*}

Repeating this procedure for a circuit $\hat C$ with the last $n-r$ capacitors removed
produces an identical SRG.
We can compute an SRG for
the error relation $C - \hat C$ by subtracting $\srg{\hat C}$ from $\srg{C}$.  This
is bounded by the disc illustrated below.

\begin{center}
        \includegraphics{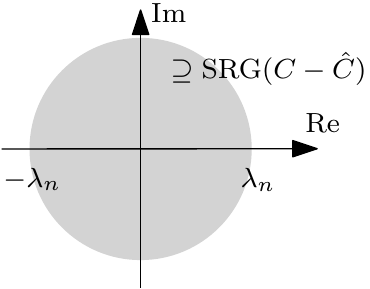}
\end{center}

It follows from Proposition~\ref{prop:finite_gain} that the error relation, which
maps $i \coloneqq i_0$ to $e \coloneqq v_0 - \hat v_0$, has an
incremental gain bound of $\lambda_n$:
\begin{IEEEeqnarray*}{rCl}
\sup_{i \in L_2, \norm{i}\neq 0} \frac{\norm{v_0 - \hat v_0}}{\norm{i}} &\leq& \sup_{i_1, i_2 \in L_2, i_1 \neq i_2}\frac{\norm{e_1 - e_2}}{\norm{i_1 - i_2 }} \leq \lambda_n. 
\end{IEEEeqnarray*}
This bound depends only on $\lambda$ and $n$, and approaches a constant as $n\to
\infty$.

Comparing the SRGs of $C$ and $\hat C$ shows that both circuits are output-strictly
incrementally passive, with an incremental secant gain of $\lambda_n$.

By way of comparison, applying the balanced truncation method presented in \autocite{Besselink2014}, with $R^{-1}(v) = \tanh(v) + v$, results in a pure
truncation of the continued fraction of the circuit, by removing $n-r$ repeated units,  
and gives an error bound 
\begin{IEEEeqnarray*}{rCl}
 \sup_{i \in L_2, \norm{i}\neq 0}\frac{\norm{v_0 - \hat v_0}}{\norm{i}}  &\leq& \frac{n - r}{(1 - \gamma)^2},
\end{IEEEeqnarray*}
where $\gamma = l \lambda$, and $l$ is the largest eigenvalue of the $n\times n$ matrix
{\small
\begin{IEEEeqnarray*}{rCl}
\begin{pmatrix}
        -2& 1 &0 &\ldots &0 \\
        1 & -2 & 1 & \ldots & 0\\
        0 & 1 & -2 & \ldots& 0\\
        \vdots & \vdots &\vdots & \ddots & \vdots\\
        0 & 0 & 0 & \ldots & 0
\end{pmatrix}.
\end{IEEEeqnarray*}
}
Note that the eigenvalue $l$ converges to $4$ as $n\to \infty$.
This bound is tighter than $\lambda_n$ for small $n$, but diverges as $n \to \infty$.
The two bounds are plotted in Figure~\ref{fig:bounds}.

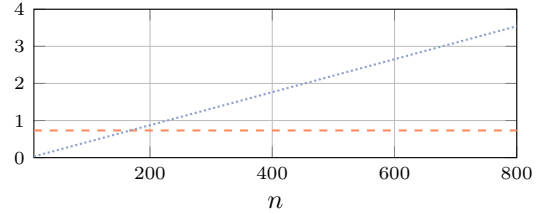
\begin{figure}[h]
        \centering
        \begin{tikzpicture}[black, every tick label/.append style={font=\scriptsize}]
        \begin{axis}[
        height = 0.2\textwidth,
        width =  0.45\textwidth, 
        xmin=10,xmax=800,
        ymin=-0,ymax=4,
        xlabel={$n$},
        grid
        ]
        \addplot [redp, thick, dashed]  table [x=n, y=s,  col
                sep=comma]{data/bounds.csv};
        \addplot [bluep, thick, densely dotted]  table [x=n, y=b,  col
                sep=comma]{data/bounds.csv};
        \end{axis}
        \end{tikzpicture}
        \vspace{-.3cm}
        \caption{Maximum modulus of $\srg{C - \hat C}$ (orange, dashed) and the error bound obtained
        from \textcite{Besselink2014} (blue, dotted).  The truncated circuit has length
$r=3$, and $\lambda=2$.}%
        \label{fig:bounds}
\end{figure}

The performance of the two truncations is compared in Figure~\ref{fig:simulation},
for an input of $i_0(t) = \sin(t)$, and Figure~\ref{fig:simulation2}, for an input of
$i_0(t) = \sin(2t)$.  Both simulations use an initial condition of $1$ V across each
capacitor.  The original circuit length is $n = 50$, and the truncated circuit length
is $r=3$.  For the method we present here, the $n-r$ nonlinear resistors which remain after
the capacitors are removed are approximated by piecewise-linear functions. 
The method we present here has lower error in both
cases, both in absolute magnitude and in phase shift.

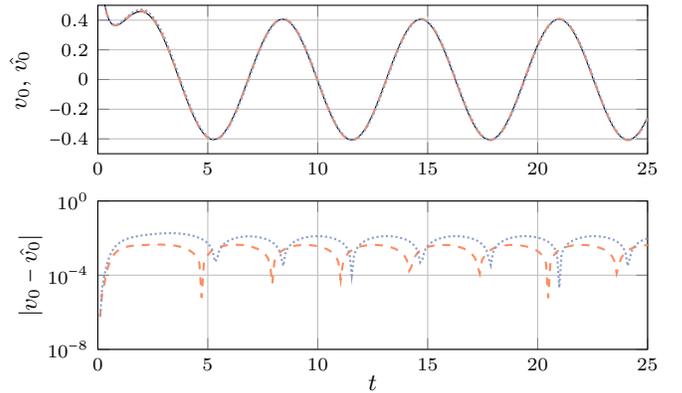
\begin{figure}[h]
\centering
\begin{tikzpicture}[black, every tick label/.append style={font=\scriptsize}]
\matrix{
\begin{axis}[
height = 0.2\textwidth,
width =  0.5\textwidth, 
xmin=0,xmax=25,
ymin=-0.5,ymax=0.5,
ylabel={{\small $v_0, \, \hat{v_0}$}},
grid
]
\addplot [black]  table [x index = {0}, y index = {1},  col
        sep=comma]{data/NL_resistor_state_system_sint.csv};
\addplot [redp,  thick, dashed]  table [x index = {0}, y index = {1},  col
        sep=comma]{data/NL_resistor_output_Tom_sint3.csv};
\addplot [bluep,  thick, densely dotted]  table [x index = {0}, y index = {1},  col
        sep=comma]{data/NL_resistor_output_Bess_sint3.csv};
\end{axis}
\\
\begin{semilogyaxis}[
height = 0.2\textwidth,
width =  0.5\textwidth, 
xmin=0,xmax=25,
ymin=10^-8,ymax=10^0,
xlabel shift = -5 pt,
ylabel shift = -5 pt,
xlabel={{\small $t$ }},
ylabel={{\small $|v_0-\hat{v_0}|$}},
grid
]
\addplot [redp,  thick, dashed]  table [x index = {0}, y index = {1},  col
        sep=comma]{data/NL_resistor__error_Tom_sint3.csv};
\addplot [bluep,  thick, densely dotted]  table [x index = {0}, y index = {1},  col
        sep=comma]{data/NL_resistor__error_Bess_sint3.csv};
\end{semilogyaxis}
\\[0 em]
};
\end{tikzpicture}
\centering
\vspace{-0.75cm}
\caption{Top: Time history of the output of the circuit in
Figure~\ref{fig:example_circuit_1} with $n=50$ (solid black) and those of the reduced
order models of order ${r=3}$ using the method we present here (dashed orange) and the
differential balanced truncation method of \autocite{Besselink2014} (dotted blue).  The
initial condition is $1$ V across each capacitor, and the input is $i_0(t) = \sin(t)$.  Bottom: Time history of the corresponding output errors in absolute value (logarithmic scale).}
\label{fig:simulation}
\end{figure}%

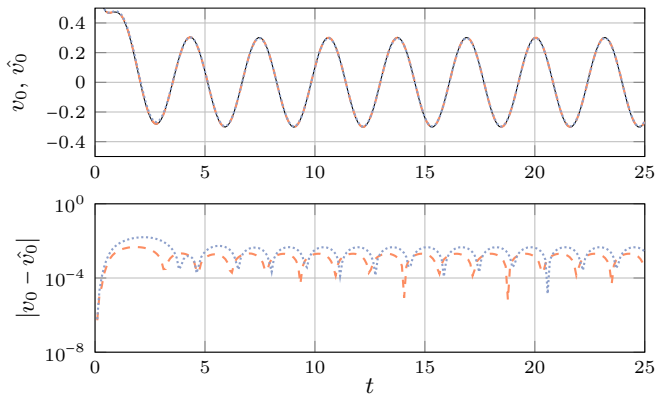
\begin{figure}[h]
\centering
\begin{tikzpicture}[black, every tick label/.append style={font=\scriptsize}]
\matrix{
\begin{axis}[
height = 0.2\textwidth,
width =  0.5\textwidth, 
xmin=0,xmax=25,
ymin=-0.5,ymax=0.5,
ylabel={{\small $v_0, \, \hat{v_0}$}},
grid
]
\addplot [black]  table [x index = {0}, y index = {1},  col
        sep=comma]{data/NL_resistor_state_system.csv};
\addplot [redp,  thick, dashed]  table [x index = {0}, y index = {1},  col
        sep=comma]{data/NL_resistor_output_Tom3.csv};
\addplot [bluep,  thick, densely dotted]  table [x index = {0}, y index = {1},  col
        sep=comma]{data/NL_resistor_output_Bess3.csv};
\end{axis}
\\
\begin{semilogyaxis}[
height = 0.2\textwidth,
width =  0.5\textwidth, 
xmin=0,xmax=25,
ymin=10^-8,ymax=10^0,
xlabel shift = -5 pt,
ylabel shift = -5 pt,
xlabel={{\small $t$ }},
ylabel={{\small $|v_0-\hat{v_0}|$}},
grid
]
\addplot [redp,  thick, dashed]  table [x index = {0}, y index = {1},  col
        sep=comma]{data/NL_resistor__error_Tom3.csv};
\addplot [bluep,  thick, densely dotted]  table [x index = {0}, y index = {1},  col
        sep=comma]{data/NL_resistor__error_Bess3.csv};
\end{semilogyaxis}
\\[0 em]
};
\end{tikzpicture}
\centering
\vspace{-0.6cm}
\caption{The same experiment as Figure~\ref{fig:simulation}, with $i_0(t) = \sin(2t)$.}
\label{fig:simulation2}
\end{figure}%

In contrast with the balanced truncation method of \autocite{Besselink2014}, 
$R$ can be non-differentiable and non-invertible (for example, a unit ideal
saturation), does not have to be a function (for example, an ideal diode) and need not be time-invariant; the element closest to the port can be linear or nonlinear; and the
voltage to current relation is just as easily analysed as the current to voltage
relation. 

\section{Conclusions} \label{sec:conclusions}

This paper explores a simple method for approximating systems which are modelled as
the port behavior of a series/parallel interconnection of nonlinear relations.
Deleting the elements furthest from the port corresponds to truncating a continued
fraction.  Resistances can be left in place and lumped into a single element.  
This procedure automatically guarantees the preservation of properties
such as incremental positivity (regular, input-strict and output-strict) and finite
incremental gain. 

The error introduced by the truncation can be evaluated using the SRGs of the
original and truncated systems.  Distances between the two SRGs correspond to errors
in quantities such as the incremental secant gain.  Furthermore, an SRG can be
computed for the error relation, and this gives a bound on the incremental gain from
the input to the truncation error.

A natural open question concerns the generality of the series/parallel structure:
when can a system be modelled as a series/parallel one-port?  
This is a nonlinear version of one of the earliest questions in circuit theory: when
can a transfer function be realised as the port behavior of an RLC one-port?  This
question arose in the work of Foster \autocite{Foster1924}, Cauer
\autocite{Cauer1926}, and Brune
\autocite{Brune1931}, and a constructive solution was provided by Bott and Duffin
\autocite{Bott1949}.  The Bott-Duffin construction is still the subject of active
research \autocite{Hughes2014, Hughes2017d}.  We leave the equivalent nonlinear
construction as a question for future research.

\printbibliography

\end{document}